\documentclass[singlecolumn,laps,pra,amsmath,amssymb]{revtex4}
\usepackage{amsfonts,amssymb,amsmath}      
\usepackage{dcolumn}
\usepackage{bm}
\usepackage{hyperref}
\usepackage{epstopdf}
\usepackage{graphicx}
\usepackage{mathtools}
\usepackage{amsthm}
\usepackage{algorithm2e}
\allowdisplaybreaks

\newcommand{\nn}{\nonumber \\}
\newcommand{\bra}[1]{\langle{#1}|}
\newcommand{\ket}[1]{|{#1}\rangle}

\def\lsim{\mathrel{\rlap{\lower4pt\hbox{$\sim$}}
    \raise1pt\hbox{$<$}}}                
\def\gsim{\mathrel{\rlap{\lower4pt\hbox{$\sim$}}
    \raise1pt\hbox{$>$}}}                
\newcommand{\thm}[1]{\hyperref[thm:#1]{Theorem~\ref*{thm:#1}}}
\newcommand{\lem}[1]{\hyperref[lem:#1]{Lemma~\ref*{lem:#1}}}

\newcommand{\Hamseg}{V}
\newcommand{\appseg}{\tilde V}
\newcommand{\oaa}{V_{\rm oaa}}
\newcommand{\err}{\Delta}
\newcommand{\corr}{V_C}
\newcommand{\corrp}{V^+_C}
\newcommand{\appcor}{\tilde V_C}
\newcommand{\ser}{W}
\newcommand{\segs}{r}
\newcommand{\obv}{OAA}

\newtheorem{theorem}{Theorem}
\newtheorem{lemma}{Lemma}

\begin{document}
\title{Corrected quantum walk for optimal Hamiltonian simulation}
\author{Dominic W. Berry${}^1$ and Leonardo Novo${}^{1,2,3}$}
\affiliation{${}^1$Department of Physics and Astronomy, Macquarie University, Sydney, NSW 2109, Australia}
\affiliation{${}^2$Physics of Information and Quantum Technologies Group, Instituto de Telecomunica\c{c}\~oes, Lisbon, Portugal}
\affiliation{${}^3$Instituto Superior T\'{e}cnico, Universidade de Lisboa, Portugal}
\date{\today}

\begin{abstract}
We describe a method to simulate Hamiltonian evolution on a quantum computer by repeatedly using a superposition of steps of a quantum walk, then applying a correction to the weightings for the numbers of steps of the quantum walk.
This correction enables us to obtain complexity which is the same as the lower bound up to double-logarithmic factors for all parameter regimes.
The scaling of the query complexity is $O\left( \tau \frac{\log\log\tau}{\log\log\log\tau} + \log(1/\epsilon) \right)$ where $\tau := t\|H\|_{\max}d$, for $\epsilon$ the allowable error, $t$ the time, $\|H\|_{\max}$ the max-norm of the Hamiltonian,  and $d$ the sparseness.
This technique should also be useful for improving the scaling of the Taylor series approach to simulation, which is relevant to applications such as quantum chemistry.
\end{abstract}
\maketitle

\section{Introduction}
The simulation of physical quantum systems is a natural application where quantum computers can achieve an exponential speedup (in the dimension of the system), and was Feynman's original motivation for proposing quantum computers \cite{Feyn}.
An algorithm for the case of a physical system composed of low-dimensional subsystems and a Hamiltonian that is a sum of interaction Hamiltonians was proposed by Lloyd \cite{lloyd}.
An alternative scenario is that where the matrix representing the Hamiltonian is sparse, and there is a procedure for calculating the positions and values of nonzero entries, which can be regarded as an oracle \cite{ats}.
The advantage of this approach is that it can be used to not just simulate Hamiltonians corresponding to physical systems, but to design other algorithms \cite{HHL,childs03,childs09,KC15}.
There have been many papers providing improved algorithms for simulating sparse Hamiltonians \cite{BACS,wiebe11,childs10,PQSV11,BC12,CW12,BCCKS14,BCCKS15,BCK15}.

One particular approach is that of a quantum walk \cite{childs10,BC12}, where a step of the quantum walk can be implemented using an oracle for the Hamiltonian, and has eigenvalues related to that of the Hamiltonian.
In Refs \cite{childs10,BC12}, the technique was to use phase estimation to estimate the eigenvalue of the quantum walk step, and use that to apply the appropriate eigenvalue for Hamiltonian evolution.
Another common approach is that based on a Lie-Trotter product \cite{lloyd,ats,BACS,wiebe11,PQSV11}.
A more recent improvement is to use a control qubit for each step in the product, and compress the control qubits \cite{BCCKS14}.
Surprisingly, this turns out to be equivalent to an approach based on a Taylor expansion of the exponential for the Hamiltonian evolution \cite{BCCKS15}.
These approaches provide a scaling that is logarithmic in $1/\epsilon$, where $\epsilon$ is the allowable error for the simulation, in contrast to previous approaches which were polynomial in $1/\epsilon$.

These results motivated an improved approach to quantum walks, where a superposition of different numbers of steps of the quantum walk is used \cite{BCK15}.
This approach enables simulations with query complexity
\begin{equation}
\label{orgsca}
O\left( \tau \frac{\log(\tau/\epsilon)}{\log\log(\tau/\epsilon)} \right),
\end{equation}
where $\tau := t\|H\|_{\max}d$.
Here $t$ is the time that the Hamiltonian evolution is to be simulated over, $\|H\|_{\max}$ is the max-norm of the Hamiltonian, and $d$ is the sparseness of the Hamiltonian, which is the maximum number of nonzero elements in any row or column.
The query complexity corresponds to the number of calls to the oracle for elements of the Hamiltonian.
In contrast, the lower bound to the complexity is \cite{BCK15}
\begin{equation}
\Omega\left( \tau + \frac{\log(1/\epsilon)}{\log\log(1/\epsilon)} \right).
\end{equation}
Hence, although there is optimal scaling (up to a logarithmic factor) in any of the parameters individually, the complexity of the algorithm has a product, whereas the lower bound has a sum, so there is room to improve the complexity in the parameter regime where $\tau$ is close to $\log(1/\epsilon)$.
Note that there is also a complexity in terms of the number of additional gates.
We do not consider that complexity here.

Here we show a method to reduce the complexity to close to the lower bound.
The idea is to use the superposition of different numbers of steps of the quantum walk as in \cite{BCK15}, but correct the weightings.
The scaling of the query complexity for our approach is
\begin{equation}
O\left( \tau \frac{\log\log\tau}{\log\log\log\tau} + \log(1/\epsilon) \right).
\end{equation}
Our technique is flexible enough that it can be applied to the Taylor series approach to quantum walks as well (though we will not provide the proof here).
The Taylor series approach is important because it can be used for situations like quantum chemistry \cite{chem1,chem2}, where the Hamiltonian naturally decomposes into a sum of terms, but it is difficult to construct an oracle to give elements of the Hamiltonian.
Another advantage of the Taylor series approach is that it gives better scaling for the number of additional gates (the above scalings are the query complexity).
After completion of this work, we were made aware of another approach that yields complexity slightly closer to the lower bound \cite{low}, but is only applicable to the quantum walk approach, not the Taylor series approach.

In the next Section we provide the theoretical background of the techniques from Ref.~\cite{BCK15}.
In Section \ref{sec:summ} we provide a description of our algorithm, and introduce definitions required in later sections.
Then in Section \ref{sec:algs} we define algorithms with a single round of correction and with two rounds of correction.
We provide the proofs for a single round of correction and two rounds of correction in Sections \ref{sec:single} and \ref{sec:double}, respectively.
We conclude in Section~\ref{conc}.

\section{Background}
\label{sec:back}
In the Hamiltonian simulation problem, a Hamiltonian $H$ acting on $n$ qubits is given, as well as a time $t$ and maximum allowable error $\epsilon>0$, and the task is to implement the unitary operation $e^{-iHt}$  within error $\epsilon$.
The error may be quantified by the diamond norm distance, although as shown in \cite{BCK15} it is equivalent to consider the spectral norm of the difference of the operators.
The question of appropriate measurements to perform on the resulting quantum state is separate from the Hamiltonian simulation problem.
For sparse Hamiltonian simulation, the Hamiltonian has a matrix representation in the computational basis with no more than $d$ nonzero entries in any row or column.
The Hamiltonian is specified by two oracles.
One takes as input a row number $j$ and integer $\ell$, and outputs the position of the $\ell$th nonzero element in row $j$.
This oracle computes the value in place, so it acts as
\begin{equation}
O_F\ket{j,\ell}=\ket{j,f(j,\ell)}\, ,
\end{equation}
where $f(j,\ell)$ is the function giving the column index.
The other oracle takes as input the row and column number for the Hamiltonian, and outputs an encoding of the value of corresponding matrix element of $H$.
This oracle acts as
\begin{equation}
O_H\ket{j,k,z} = \ket{j,k,z\oplus H_{jk}}\, ,
\end{equation}
where $z$ is a bit string used to represent entries of $H$, and $\oplus$ indicates bitwise addition.
These oracles represent some procedure for calculating the positions and values of nonzero entries in the Hamiltonian.

Next we summarise the key results needed from \cite{childs10,BC12,BCK15}.
Those works consider a quantum walk which is based on controlled state preparation.
The Hilbert space is expanded from ${\mathbb C}^{\cal N}$, where ${\cal N}=2^n$, to ${\mathbb C}^{2{\cal N}}\otimes {\mathbb C}^{2{\cal N}}$.
First an ancilla qubit is appended, which expands the space to ${\mathbb C}^{2{\cal N}}$, then the original Hilbert space and ancilla qubit are duplicated.
The duplication is performed using a controlled state preparation operator \cite{BCK15}
\begin{equation}\label{Tdef}
T := \sum_{j=1}^{\cal N} \sum_{b\in\{0,1\}} \left( \ket{j}\bra{j}\otimes \ket{b}\bra{b}\right)\otimes \ket{\varphi_{jb}}\, ,
\end{equation}
where $\ket{\varphi_{j1}}=\ket{0}\ket{1}$ and
\begin{equation}\label{varphidef}
\ket{\varphi_{j0}} = \frac 1{\sqrt{d}} \sum_{\ell\in F_j} \ket{\ell}\left[ \sqrt{\frac{H^*_{j\ell}}{X}} \ket{0} + \sqrt{1-\frac{|H_{j\ell}|}{X}}\ket{1} \right],
\end{equation}
where $F_j$ is the set of indices given by $O_F$ on input $j$, and $X\ge\|H\|_{\max}$.
For the algorithms described here, we will just take $X = \|H\|_{\max}$.
The method to perform the controlled state preparation is described in Lemma 4 of Ref.~\cite{BC12}.
The key feature of it that we need here is that it uses $O(1)$ calls to the oracles $O_F$ and $O_H$.

We use the convention that the basis states are given in the order $\ket{j_1}\ket{b_1}\ket{j_2}\ket{b_2}$, where $\ket{j_1}$ and $\ket{j_2}$ are the original Hilbert space and duplicated space, and $\ket{b_1}$ and $\ket{b_2}$ are the orignal ancilla qubit and duplicated ancilla qubit.
The quantum walk step is constructed from this controlled state preparation via
\begin{equation}\label{Udef}
U := iS(2TT^\dagger -\openone)\, ,
\end{equation}
where $S$ swaps the registers, as $S\ket{j_1}\ket{b_1}\ket{j_2}\ket{b_2}=\ket{j_2}\ket{b_2}\ket{j_1}\ket{b_1}$.
All algorithms of this type work by starting in the initial Hilbert space, performing $T$ to map the state into two copies of the Hilbert space, performing steps of the quantum walk $U$ on these two subsystems, then using $T^\dagger$ to map the system back to the original Hilbert space.

The crucial feature of the step $U$ is that its eigenvalues and eigenstates are related to those of $H$.
In particular, an initial eigenstate $\ket{\lambda}$ of $H$ will be mapped under $T$ to a superposition of two eigenstates of $U$, $\ket{\mu_+}$ and $\ket{\mu_-}$.
If the eigenvalue of the Hamiltonian is $\lambda$, the corresponding eigenvalues of $U$ are
\begin{equation}
\mu_\pm = \pm e^{\pm i\arcsin(\lambda/Xd)}.
\end{equation}
Using this expression and the generating function for Bessel functions, it can be shown that \cite{BCK15}
\begin{equation}\label{bessels}
\sum_{m=-\infty}^{\infty} J_m(-tXd) \mu_\pm^m = e^{-i\lambda t}.
\end{equation}
This expression effectively means that there is the equivalence
\begin{equation}
\sum_{m=-\infty}^{\infty} J_m(-tXd) U^m \equiv e^{-iHt},
\end{equation}
so if it is possible to apply a sum of powers of the quantum walk step $U$, it is possible to effectively apply the evolution under the Hamiltonian.
Note that the Hamiltonian acts on a different space than $U$, so these operators are not equal.
One needs to map from the original Hilbert space to the two copies of the space, apply the superposition of powers of $U$, then map back to the original Hilbert space to obtain $e^{-iHt}$.

The method to apply a linear combination of unitaries of this form is described in detail in Section IIIB of Ref.~\cite{BCK15}.
We summarise it for the case of powers of $U$ here.
First, rather than attempting to obtain the evolution for the complete time $t$, we divide the time up into $r$ segments, and apply the linear combination of unitaries for each time $t/r$.
Second, the infinite sum in Eq.~\eqref{bessels} needs to be truncated to finite values.
We denote the cutoff by $M$.
The cutoff needs to be chosen sufficiently large that the error for the complete simulation due to the discarded terms is no greater than $\epsilon$.
If we were to attempt to apply the sum of unitaries given in Eq.~\eqref{bessels}, but for time $t/r$, the weightings would be $v_m = J_m(-tXd/r)$.
The operation that we wish to perform can then be written as
\begin{equation}\label{defaps}
\appseg = \sum_{m=-M}^{M} v_m U^m .
\end{equation}

To perform the linear combination of unitaries, first an ancilla is prepared in the state
\begin{equation}\label{chidef}
\ket{\chi} = \frac 1{\sqrt{s}} \sum_{m=-M}^{M} \sqrt{v_m} \ket{m}\, ,
\end{equation}
where $s=\sum_{m=-M}^M |v_m|$.
Preparation of this ancilla requires no queries, and only $O(M)$ gates.
Then one performs a controlled operation
\begin{equation}\label{controlled}
{\rm select}(U) := \sum_{m=-M}^{M} \ket{m}\bra{m} \otimes U^m.
\end{equation}
Then, given an initial state $\ket{\psi}$, the state after appending the ancilla in the state $\ket{\chi}$ and applying ${\rm select}(U)$ is
\begin{equation}\label{intst}
{\rm select}(U) \ket{\chi}\ket{\psi} = \frac 1{\sqrt{s}}\sum_{m=-M}^{M} \sqrt{v_m} \ket{m} U^m\ket{\psi}\, .
\end{equation}
Next, if one were to apply the projector $\ket{\chi}\bra{\chi}$, corresponding to measuring the ancilla as being in the state $\ket{\chi}$, the resulting state would be
\begin{equation}
\ket{\chi}\frac 1s \appseg \ket{\psi}\, ,
\end{equation}
where the normalization indicates the probability of success.
That is, the desired sum of unitaries has been performed, but the probability of success is only about $1/s^2$.

An alternative way of viewing the measurement is that one first inverts the state preparation on the ancilla.
That will then yield the state
\begin{equation}\label{beforeoaa}
\ket{0} \otimes \frac 1s \appseg \ket{\psi} + \ket{\perp}\, ,
\end{equation}
where $\ket{0}$ is the zero initial state for the ancilla, and $\ket{\perp}$ is a possibly entangled state between the ancilla and target system that is perpendicular to $\ket{0}$ on the ancilla.
If one were to perform a measurement on the ancilla and obtain the result $0$, that would correspond to applying the projector $\ket{\chi}\bra{\chi}$ to the state in Eq.~\eqref{intst}.

Instead of just measuring, one can perform a procedure called \emph{oblivious amplitude amplification} ({\obv}) to boost the amplitude of the first term in Eq.~\eqref{beforeoaa} to close to $1$.
Define $W_\chi$ to be the operation corresponding to preparing the state $\ket{\chi}$ from $\ket{0}$, then applying ${\rm select}(U)$, then inverting the state preparation operation.
See Lemma 5 of Ref.~\cite{BCK15} for an explanation of this operation (the notation used in \cite{BCK15} is $W$ instead of $W_\chi$).
Also define $P:= \ket{0}\bra{0}\otimes \openone$ to be the projection onto $\ket{0}$ on the ancilla.
Then in a single step of \obv, one reflects about the zero state on the ancilla via $\openone-2P$, performs $W_\chi^\dagger$, reflects about the zero state on the ancilla again, and performs $W_\chi$.
That is, the total operation performed for the segment is
\begin{equation}
-W_\chi (\openone-2P) W_\chi^\dagger (\openone-2P) W_\chi \, ,
\end{equation}
where the minus sign avoids a global phase factor.
See Section~3.2 of Ref.~\cite{BCK15} for a detailed explanation of OAA and the proof of accuracy given that the sum of unitaries is not exactly unitary.
It is also possible to consider multiple steps of {\obv}, but here we only need a single step.
The key feature of {\obv} that we need here is that after one step it gives the result
\begin{equation}\label{afteroaa}
\ket{0} \otimes \left(\frac 3s \appseg - \frac 4{s^3}\appseg \appseg^\dagger \appseg \right)\ket{\psi} + \ket{\perp'}\, ,
\end{equation}
where $\ket{\perp'}$ is some new perpendicular state.
If $s$ is exactly $2$, and $\appseg$ is exactly unitary, then the state would be $\ket{0} \appseg\ket{\psi}$, so the desired sum of unitaries has been performed.
Even if $\appseg$ is not exactly unitary, then the result will still be close to that desired.
Also, if $s$ is not exactly $2$, but is instead less than $2$, then it is trivial to include an extra ancilla that effectively increases it to $2$, and the result is the same \cite{BCK15}.

As the argument of the Bessel function is increased, it is necessary to increase $M$ in order to limit the error.
Then the value of $s$ also increases (although not monotonically).
This is why the evolution is divided up into $r$ segments.
By choosing $r$ sufficiently large, $s\le 2$, so only a single step of OAA is required for each segment.
Then the overall Hamiltonian evolution is simulated by performing the simulation for each time $t/r$ in succession.
The overall complexity is then a product of the number of segments times the complexity for each segment.
The value $s\le 2$ can be obtained with the argument of the Bessel function $O(1)$.
The argument of the Bessel function is $-t \, Xd/r$, so the total number of segments should be $r=O(tXd)$.
Taking $X=\|H\|_{\max}$ and $\tau := t\| H \|_{\max} d$, the number of segments is $O(\tau)$.

In Ref.~\cite{BCK15} the segments are just simulated in succession with no correction.
Therefore, to bound the final error as no larger than $\epsilon$, the error for each segment should be no larger than $O(\epsilon/\tau)$.
The values of Bessel functions are bounded as \cite{NIST}
\begin{equation}
|J_m(z)| \le \frac 1{m!} \left|\frac z2\right|^{|m|}.
\end{equation}
The error in the truncation will correspond to the absolute values of the Bessel functions which are omitted.
If we truncate at $M>0$, then the truncated terms will be bounded as \cite{BCK15}
\begin{equation}
2 \sum_{m=M+1}^{\infty} |J_m(z)| \le 4\frac 1{(M+1)!} \left|\frac {z}2\right|^{M+1}.
\end{equation}
Using this expression, the error for the segment can be limited to $O(\epsilon/\tau)$ if the scaling of the cutoff is
\begin{equation}
M = O\left( \frac{\log(\tau/\epsilon)}{\log\log(\tau/\epsilon)} \right).
\end{equation}
Recall that each application of $U$ uses $O(1)$ queries, and $M$ is the cutoff to the maximum power of $U$, so
the complexity for a single segment is $O(M)$.
In Ref.~\cite{BCK15} the overall complexity is then obtained by multiplying by the number of segments to give $O(\tau M)$, which gives the complexity scaling in Eq.~\eqref{orgsca}.
\section{New Algorithms with Corrections}
\label{sec:summ}
The basis of our new technique is to allow a larger error for each segment, then reduce the error by performing a correction at the end.
Because the simulation of the individual segments need not be as accurate, the cutoff on the Bessel functions can be smaller and need not depend on $\epsilon$.
That eliminates the multiplying factor depending on $\epsilon$ in the complexity.
The correction at the end can use powers of $U$ corresponding to the total complexity of the simulation thus far without changing the complexity except by a constant factor.
Therefore the correction can greatly reduce the error.

The reason this approach works is that one step of {\obv} gives an effective operation which is (taking $s=2$) \cite{BCCKS15}
\begin{equation}\label{oaadef}
\oaa := \frac 32 \appseg - \frac 12 \appseg \appseg^\dagger \appseg ,
\end{equation}
where $\appseg$ was defined in Eq.~\eqref{defaps}, and is the sum of powers of $U$ used to approximate the Hamiltonian evolution for a single segment.
What this means is that after performing the {\obv} we have still performed an operation that is a sum of powers of $U$.
To be more specific, the state has success flagged by zero in the ancilla, as well as a failure component, so is
\begin{equation}\label{voaa}
\ket{0} \otimes \oaa \ket{\psi} + \ket{\perp'}\, .
\end{equation}
Also, the operation $\appseg$ is
\begin{equation}\label{apseg}
\appseg := \sum_{m=-M}^{M} J_m(z) U^m ,
\end{equation}
where $z$ satisfies
\begin{equation}\label{zrestrict}
\sum_{m=-M}^M |J_m(z)| \le 2\, .
\end{equation}
For a single round of error correction, we wish to have $zr=-\tau$.
If we select $z$ to be the maximum value such that Eq.~\eqref{zrestrict} is satisfied, then we would only be able to obtain a discrete set of values of $\tau$ (and therefore $t$).
We instead take $z$ to be the maximum value satisfying both \eqref{zrestrict} and $zr=-\tau$ for some integer $r$.

The effective operation after many segments is then a power of $\oaa$, with success flagged by the $\ket{0}$ states in ancillas.
That is, after $r$ segments the state becomes
\begin{equation}\label{voaar}
\ket{0}^{\otimes r} \otimes \oaa^r \ket{\psi} + \ket{\perp_2}\, .
\end{equation}
Hence, the effective operation is still a sum of powers of $U$.
This sum of powers of $U$ is something that we can easily calculate, and we can compare the weightings to those we would want in order to exactly simulate the Hamiltonian evolution.
At the end we can apply another operation that is a sum of powers of $U$ to correct the weightings, and reduce the error to order $\epsilon$.

If the truncation of this sum of powers is of the same order as the maximum power of $U$ for the operations that have been performed so far, then this correction will only give a constant multiplying factor to the complexity.
What is more, this correction will have most of its weight on $m=0$ (the identity), which means that the amplitude on success is high.
More explicitly, define $\corr$ to be the operation such that
\begin{equation}\label{cordef}
\corr\oaa^r = \sum_{m=-\infty}^{\infty} J_m(z r) U^m .
\end{equation}
That is, it yields the desired weightings to give the evolution under the Hamiltonian.
The correction has an expansion in terms of powers of $U$ as
\begin{equation}
\corr = \sum_{m=-\infty}^{\infty} a_m U^m,
\end{equation}
which we use to define the coefficients $a_m$.
In practice we need to truncate the sum to limit the complexity of the correction, giving
\begin{equation}\label{apcordef}
\appcor = \sum_{m=-N}^{N} a_m U^m.
\end{equation}
Applying this approximate correction via an ancilla, in the same way as described above for $\appseg$,  gives the state
\begin{equation}\label{aftercor}
\ket{0}^{\otimes (r+1)} \otimes \frac 1s \appcor\oaa^r \ket{\psi} + \ket{\perp_3}\, ,
\end{equation}
where $s=\sum_{m=-N}^N |a_m|$.
Again, provided that $s<2$, {\obv} can be used to obtain a result close to $\appcor\oaa^r \ket{\psi}$.
The operation after {\obv}, denoted $\oaa'$, is given by
\begin{equation}
\oaa' = \frac 32 \appcor\oaa^r - \frac 12 \appcor\oaa^r (\appcor\oaa^r)^\dagger \appcor\oaa^r\, .
\end{equation}
Again, the resulting state is flagged by zero states in the ancillas, and there is a contribution from an error state, in the same way as in Eqs.~\eqref{beforeoaa}, \eqref{afteroaa}, \eqref{voaa}, \eqref{voaar}, and \eqref{aftercor}.
From now on we will not write these explicitly for brevity.
Using Lemma 6 of \cite{BCK15}, if $\appcor\oaa^r$ is within $\delta$ of a unitary matrix, then $\oaa'$ is within $\delta$ of $\appcor\oaa^r$.

In order to obtain $s\le 2$ for the correction, we need the error in the segments to not build up above $O(1)$.
This means that the error in the individual segments needs to be $O(1/\tau)$, so the cutoff is
\begin{equation}
M = O\left( \frac{\log\tau}{\log\log\tau} \right).
\end{equation}
Multiplying by the number of segments, this gives a factor in the complexity of
\begin{equation}
O\left(\tau \frac{\log\tau}{\log\log\tau} \right).
\end{equation}
Another limiting factor is that the final error needs to be $O(\epsilon)$.
In the case where $\epsilon$ is small, it may be necessary to increase the cutoff $N$ for the final correction.
It is also convenient to increase $M$, though $M$ does not need to scale with $\epsilon$.
It turns out that it is sufficient to choose $M\tau = O(\log(1/\epsilon))$,
which gives the $O(\log(1/\epsilon))$ term in the final scaling.

The discussion so far is for a single round of error correction, which is addressed in detail in Section \ref{sec:single}.
It is also possible to repeatedly perform error correction, then perform a final round of error correction at the end.
That is, we have the operation $\oaa'$ repeated $r'$ times, with success again flagged by an ancilla state of $\ket{0}$.
Then the new exact correction operation $\corr'$, is such that
\begin{equation}
\corr' (\oaa')^{r'} = \sum_{m=-\infty}^{\infty} J_m(z r r') U^m .
\end{equation}
Note that we now have $r$ multiplied by $r'$, so we should have $\tau = -z r r'$.
For this case we need to select $z$ such that Eq.~\eqref{zrestrict} is satisfied and $\tau = -z r r'$, for appropriate choices of $r$ and $r'$.
The exact correction operation can again be expressed as a sum over powers of $U$ as
\begin{equation}
\corr'  = \sum_{m=-\infty}^{\infty} a'_m U^m ,
\end{equation}
which we use to define $a'_m$.
The approximate correction operation is then the truncated form of this, with
\begin{equation}
\appcor'  = \sum_{m=-N'}^{N'} a'_m U^m .
\end{equation}
This correction operation is applied using control registers and ${\rm select}(U)$ operations, then {\obv} is again applied.

Again we can estimate the complexity via a relatively simple approach.
Take $r=\Theta(\log\tau)$, so the error after the first round of error correction (i.e., the error in $\oaa'$) is $O(1/\tau)$.
Then we can take $r'=\Theta(\tau/\log \tau)$ to obtain $\tau = -z r r'$ as desired.
With this value of $r'$ the error will still be $O(1)$, so the final correction can be performed.
Moreover, taking $r=\Theta(\log\tau)$, we can allow error for the individual segments (i.e., in $\appseg$) to be $O(1/\log\tau)$.
That means we can take the cutoff for $\appseg$ to be
\begin{equation}
M = O\left( \frac{\log\log\tau}{\log\log\log\tau} \right).
\end{equation}
That is what gives our double-log factor in the final result.
The details of this double round of error correction are given in Section \ref{sec:double}.
\section{Summary of Algorithms}\label{sec:algs}
~\vspace{-1cm}
\RestyleAlgo{boxruled}
\LinesNumbered
\begin{algorithm}[h!]
\caption{Hamiltonian simulation with a single round of correction \label{alg1}}
\begin{changemargin}{-0.5cm}{0.8cm} 
\begin{enumerate}
\item
Apply the controlled state preparation operator $T$ defined in Eq.~\eqref{Tdef} to map the state to two copies of the Hilbert space, as well as two ancilla qubits.
The states $\ket{\varphi_{jb}}$ required for $T$ are given via $\ket{\varphi_{j1}}=\ket{0}\ket{1}$ and Eq.~\eqref{varphidef}.
\item
For segments $1$ to $r$ with $r\in \Theta(\tau)$, perform the following steps.
\begin{enumerate}
\item Append an ancilla of dimension $2M+1$ in state $\ket{0}$, and apply the operation to prepare the state $\ket{\chi}$ given in Eq.~\eqref{chidef}.
\item Perform the controlled unitary operation ${\rm select}(U)$ given in Eq.~\eqref{controlled}.
The unitary operation $U$ is defined in Eq.~\eqref{Udef}, and uses a swap operation $S$ as well as the controlled preparation $T$. 
\item Apply the inverse of the state preparation operation, giving the state in Eq.~\eqref{beforeoaa}.
That is, the operation $\appseg$ has success flagged by a $\ket{0}$ state in the ancilla.
\item Apply a single step of {\obv}, as described in Lemma 5 of Ref.~\cite{BCK15}.
This yields the operation $\oaa$, as defined in Eq.~\eqref{oaadef}, with success flagged by a zero in the ancilla.
\end{enumerate}
\item Apply the correction $\appcor$ defined in Eq.~\eqref{apcordef} via an ancilla and the unitary ${\rm select}(U)$, to give the state in Eq.~\eqref{aftercor}.
\item Apply a single step of {\obv} on steps 2 and 3 above.
\item Invert the controlled state preparation operator $T$ to map the state back into the original Hilbert space.
\end{enumerate}
\end{changemargin}
\end{algorithm}

\noindent
For the algorithm, we are provided an initial state $\ket{\psi}$ encoded on the qubits of the quantum computer.
The complete algorithm with a single round of correction is given as Algorithm \ref{alg1} above.
At the end of this procedure, we have a state that is an approximation of $e^{-iHt}\ket{\psi}$.
Provided $M$ and $N$ (the cutoffs for $\appseg$ and $\appcor$, respectively) are chosen appropriately, then the error will be within $\epsilon$.
For this algorithm we will always choose $M\ge 2$.

It is convenient for the discussion to refer to parts (a) to (d) of step 2 of Algorithm \ref{alg1} as a \emph{segment}.
Then, we refer to the procedure in steps 2 to 4 of Algorithm \ref{alg1} as a \emph{compound segment}.
The algorithm for two rounds of error correction repeats this compound segment a number of times denoted $r'$.
The cutoff $N$ for the correction $\appcor$ is taken to be $3rM$, so it has the same complexity as performing $\oaa^r$.
Then a correction denoted $\corr'$ is performed.
The complete procedure is given as Algorithm \ref{alg2} below.

\begin{algorithm}[h]
\caption{Hamiltonian simulation with two rounds of correction \label{alg2}}
\begin{changemargin}{-0.5cm}{0.8cm} 
\begin{enumerate}
\item
Apply the controlled state preparation operator $T$ defined in Eq.~\eqref{Tdef} to map the state to two copies of the Hilbert space, as well as two ancilla qubits.
The states $\ket{\varphi_{jb}}$ required for $T$ are given via $\ket{\varphi_{j1}}=\ket{0}\ket{1}$ and Eq.~\eqref{varphidef}.
\item
For compound segments $1$ to $r'$ with $r'\in \Theta(\tau/\log\tau)$, perform the following steps.
\begin{enumerate}
\item Use step 2 of Algorithm \ref{alg1} in order to apply $\oaa^r$ with $r\in \Theta(\log\tau)$.
\item Apply the correction $\appcor$ via an ancilla and the controlled unitary ${\rm select}(U)$.
\item Apply a single step of {\obv} on steps 2 and 3 above.
\end{enumerate}
\item Apply the correction $\appcor'$ via an ancilla and the controlled unitary ${\rm select}(U)$.
\item Apply a single step of {\obv} on steps 2 and 3 above.
\item Invert the controlled state preparation operator $T$ to map the state back into the original Hilbert space.
\end{enumerate}
\end{changemargin}
\end{algorithm}
\section{A Single Round of Error Correction}
\label{sec:single}
The specific results we derive for a single round of correction are as in the following theorem.
\begin{theorem}\label{thm:corthm}
A $d$-sparse Hamiltonian $H$ acting on $n$ qubits can be simulated for time $t$ within error $\epsilon>0$ with
\begin{equation}
O\left(\tau \frac{\log\tau}{\log\log\tau} + \log(1/\epsilon) \right)
\end{equation}
queries, where $\tau := t\|H\|_{\max}d$.
\end{theorem}

To compare this result with the results of \cite{BCK15}, the algorithm in that work had complexity
\begin{equation}
O\left(\tau \frac{\log(\tau/\epsilon)}{\log\log(\tau/\epsilon)} \right),
\end{equation}
and a lower bound on the complexity
\begin{equation}
O\left(\tau + \frac{\log(1/\epsilon)}{\log\log(1/\epsilon)} \right).
\end{equation}
These bounds differ most significantly in the regime where $\tau$ is comparable to $\log(1/\epsilon)$.
In this regime the previous algorithm gives complexity approximately $O(\tau^2\log\tau)$,
whereas the complexity from \thm{corthm} scales as $\tau\log\tau$, which is close to a square root improvement.

In order to prove this result, we first prove a number of intermediate lemmas.
In these proofs we make a slight change from the notation used in the previous sections.
We will now take the operations $\appseg$, $\corr$, $\appcor$, etc., explicitly as functions of $U$.
This is because we wish to calculate $s$ values for implementing a range of linear combinations of unitaries, and it is convenient to express $s$ as a functional of the operation.
In particular, for a function of the form $F(x)=\sum_n F_n x^n$, we define the functional $s$ by
\begin{equation}\label{sdef}
s(F):=\sum_n |F_n|\, .
\end{equation}

First we provide a lemma on the elementary properties of the functional $s$.
\begin{lemma}
\label{lem:s}
The functional $s$ satisfies the properties, for functions of $U$ denoted $F$ and $G$ and $\alpha$ and $\beta$ scalars,
\begin{align}\label{eq:res1}
s(\alpha F +\beta G) &\le |\alpha|s(F) + |\beta|s(G) \\
\label{eq:res2}
s(FG) &\le s(F)s(G)\, .
\end{align}
In addition, if $F$ and $G$ are sums over disjoint sets of powers, then
\begin{equation}\label{eq:res3}
s(\alpha F +\beta G) = |\alpha|s(F) + |\beta|s(G)\, .
\end{equation}
\end{lemma}
\begin{proof}{
To show the first result \eqref{eq:res1}, for
\begin{equation}
F(x)=\sum_n F_n x^n, \qquad G(x) = \sum_n G_n x^n,
\end{equation}
we have
\begin{align}
s(\alpha F +\beta G) &= \sum_n |\alpha F_n + \beta G_n| \nn
& \le |\alpha|\sum_n  |F_n| + \beta \sum_n |G_n| \nn
& = |\alpha|s(F) + |\beta|s(G)\, .
\end{align}
To show Eq.~\eqref{eq:res2},
\begin{equation}
(FG)(x) = \sum_n \sum_k F_{n-k} G_k x^n,
\end{equation}
so
\begin{align}
s(FG) &= \sum_n \left|\sum_k F_{n-k} G_k \right| \nn
&\le \sum_n \sum_k |F_{n-k}| \times |G_k| \nn
&= \sum_n |F_n| \sum_m |G_m| \nn
&= s(F)s(G)\, .
\end{align}
To show Eq.~\eqref{eq:res3}, denote the sets of powers for $F$ and $G$ by $S_F$ and $S_G$, respectively, so that $S_F\cap S_G=\emptyset$.
Then
\begin{align}
s(\alpha F +\beta G) &= \sum_n |\alpha F_n + \beta G_n| \nn
& = \sum_{n\in S_F} |\alpha F_n| + \sum_{n\in S_G} |\beta G_n| \nn
&= |\alpha|\sum_n  |F_n| + \beta \sum_n |G_n| \nn
& = |\alpha|s(F) + |\beta|s(G)\, .
\end{align}}
\end{proof}

The next lemma shows that the absolute values of the coefficients for the correction may be bounded.
This result is needed in order to show that the correction only needs a single step of \obv.
\begin{lemma}
\label{lem:sumbnd}
In Algorithm \ref{alg1}, the coefficients $a_m$ required for the correction operator $\corr(U)$ satisfy
\begin{equation}
\label{sumbnd2}
\sum_{m=-\infty}^{\infty} |a_m| \lsim \left( 1-2\sum_{|m|>M}|J_m(z)| \right)^{-\segs} ,
\end{equation}
where $\lsim$ indicates that higher-order terms in $\sum_{|m|>M}|J_m(z)|$ have been omitted.
\end{lemma}

\begin{proof}{
For each segment we use a value $z$ and a cutoff $M$, so the operation we attempt to perform is
\begin{equation}
\appseg(U) := \sum_{m=-M}^M J_m(z) U^m .
\end{equation}
This is equivalent to Eq.~\eqref{apseg}, except we have given $\appseg$ as a function of $U$.
The choice of $z$ is restricted by Eq.~\eqref{zrestrict}, which ensures that $s(\appseg)\le 2$.
Similarly, we will write the operation that provides the exact evolution for that time interval as the function
\begin{equation}
\Hamseg(U) = \sum_{m=-\infty}^\infty J_m(z) U^m .
\end{equation}
Because $s(\appseg)\le 2$, using a single step of \obv, the operation is $\oaa(U)$, with
\begin{equation}\label{oaadef2}
\oaa := \frac 32 \appseg - \frac 12 \appseg \appseg^\dagger \appseg\, .
\end{equation}
It is also convenient to define
\begin{equation}
\err := \Hamseg-\appseg\, .
\end{equation}

After performing $\oaa(U)$ a number of times given by $\segs$, the actual operation is $[\oaa(U)]^\segs$, but the desired operation is $[\Hamseg(U)]^\segs$.
The correction $\corr$ that will yield the exact operation is given implicitly by 
\begin{equation}
\Hamseg^\segs = \corr \oaa^\segs\, ,
\end{equation}
which is equivalent to Eq.~\eqref{cordef}.
This means that we must have
\begin{align}\label{correx}
\corr &= (\Hamseg^\dagger \oaa)^{-\segs} \nn
&= \left( \frac 32 \Hamseg^\dagger \appseg - \frac 12 \Hamseg^\dagger \appseg (\Hamseg^\dagger \appseg)^\dagger \Hamseg^\dagger \appseg \right)^{-\segs} \nn
&= \left( \openone  - \ser \right)^{-\segs}\nn
&= \sum_{k=0}^{\infty}\binom{\segs+k-1}{\segs-1} \ser^k,
\end{align}
where
\begin{align}
\ser &=\frac 12 \left(V^\dagger\Delta-\Delta^\dagger \appseg+ V^\dagger\Delta V^\dagger\Delta+V^\dagger\Delta\Delta^\dagger \appseg\right) \label{shorter} \\
&=\frac 12 \left[\appseg^\dagger\Delta-\Delta^\dagger \appseg+\Delta^\dagger\Delta+ (\appseg^\dagger)^2\Delta^2+(\Delta^\dagger)^2\Delta^2+2\appseg^\dagger\Delta^\dagger\Delta^2+\appseg^\dagger \appseg\Delta\Delta^\dagger + \appseg\Delta(\Delta^\dagger)^2 \right]. \label{ampamp}
\end{align}

Using the expression \eqref{ampamp} for $W$, $s(\appseg)\le 2$, and the properties of the functional $s$ in Lemma \ref{lem:s}, we obtain
\begin{equation}
s(\ser) \le 2s(\err) + 9[s(\err)]^2+6[s(\err)]^3+[s(\err)]^4\, .
\end{equation}
Hence we have
\begin{equation}
s(\corr) \le \left[ 1  - 2 s(\err) \right]^{-\segs}+ O\left(\segs \left[s(\err)\right]^2\right).
\end{equation}
We can evaluate $s(\err)$ as
\begin{equation}\label{serr}
s(\err) = s(\Hamseg-\appseg) = \sum_{|m|>M} |J_m(z)|\, .
\end{equation}
Using this expression we obtain the result \eqref{sumbnd2} required for the Lemma.}
\end{proof}

In reality we will perform a truncated operator $\appcor(U)$ which has the higher powers of $U$ truncated, but from the definition of $s$ in Eq.~\eqref{sdef}, $s(\appcor)\le s(\corr)$.
Therefore the result in Lemma~\ref{lem:sumbnd} means that the value of $\segs$ such that the sum over $|a_m|$ is approximately bounded by $2$ is
\begin{equation}\label{Mval}
\segs \approx \frac{\log 2}{2 \sum_{|m|>M} |J_m(z)|}\,  .
\end{equation}

There is an important symmetry of the coefficients that appear in all stages of the algorithms.
\begin{lemma}\label{lem:symm}
At all stages in Algorithms \ref{alg1} and \ref{alg2} the operation that is performed corresponds to a sum of powers of $U$ of the form
$\sum_n c_n U^n$ where $c_{-n}=(-1)^n c_n$.
\end{lemma}

\begin{proof}{
This symmetry holds for $\appseg$ and $\Hamseg$ due to the properties of Bessel functions \cite{NIST}.
Then all operations that are performed can be obtained from sums and products of $\appseg$ and $\Hamseg$, as well as by truncating higher-order terms.
For example, Eqs.~\eqref{correx} and \eqref{ampamp} are used for $\corr$.
Then $\appcor$ can be obtained from $\corr$ by omitting higher-order terms.

It is obvious that all sums of functions with this symmetry retain the symmetry.
It is also obvious that omitting higher-order terms retains the symmetry.
It can also be shown that products of functions with this symmetry retain the symmetry via
\begin{align}
\sum_m c_m U^m \sum_k d_k U^k &= \sum_n \left( \sum_m c_m d_{n-m}\right) U^n \nn
&= \sum_n \left( \sum_m (-1)^m c_{-m} (-1)^{n-m} d_{-(n-m)}\right) U^n \nn
&= \sum_n (-1)^n \left( \sum_m c_{-m} d_{-(n-m)}\right) U^n \nn
&= \sum_n (-1)^n \left( \sum_m c_{m} d_{n-m}\right) U^{-n},
\end{align}
where in the last line the substitution $m\mapsto -m$ and $n\mapsto -n$ has been made.
Because all operations used in the algorithms are obtained from $\appseg$ and $\Hamseg$ by procedures that retain the symmetry, the symmetry is retained in the algorithms.}
\end{proof}

Next we prove a lemma bounding the error due to truncating the superposition for the correction.
\begin{lemma}\label{lem:corerr}
In Algorithm \ref{alg1}, the coefficients $a_m$ required for the correction operator $\corr(U)$ satisfy
\begin{equation}
\label{errbnd}
\sum_{|m|>N} |a_m| \lsim {2^{\segs+1}}{\left(\frac{z\zeta}{M}\right)^{N+1}},
\end{equation}
where $\zeta\approx 1.8$ is the solution of $e^{1+1/2\zeta}=2\zeta$, and $N$ is a non-negative integer.
\end{lemma}

\begin{proof}{
A trick to bound the sum over $|a_m|$ is to define a modified function $\corrp$, which is the same as $\corr$ except with the absolute values of the coefficients,
and consider $\corrp(x)$, where $x$ is real and positive.
Using Lemma \ref{lem:symm} to give $|a_m|=|a_{-m}|$, and assuming that $x\ge 1$, we obtain 
\begin{align}\label{trick}
\sum_{|m|>N} |a_m| & = 2\sum_{m=N+1}^{\infty}|a_m| \nn
&\le \frac 2{x^{N+1}}\sum_{m=N+1}^{\infty} |a_m| x^m \nn
&\le \frac 2{x^{N+1}}\corrp(x) \, ,
\end{align}
where $\corrp(x):=\sum_{m=-\infty}^{\infty} |a_m| x^m$.
Note that in this expression we can take $N$ to be any integer $\ge 0$.
To upper bound $\corrp(x)$ we can use the expression \eqref{correx} together with \eqref{shorter}, and take the absolute values of all coefficients in Eq.~\eqref{shorter}.
That is,
\begin{align}\label{crpbnd}
\corrp(x) &\le \sum_{k=0}^\infty \binom{\segs+k-1}{\segs-1} [\ser_+(x)]^k \nn
&= \left[ 1-\ser_+(x)\right]^{-\segs} ,
\end{align}
where $\ser_+$ is the function $\ser$ modified to take the absolute values of all coefficients.

Using Eq.~\eqref{shorter}, we obtain
\begin{align}
\ser_+(x)&\le \frac 12 \left[ \sum_{q=-\infty}^{\infty} |J_q(z)| x^{-q} \sum_{|n|>M} |J_n(z)| x^n + \sum_{|n|>M}  |J_n(z)| x^{-n} \sum_{q=-M}^{M} |J_q(z)| x^{q} \right. \nn
&\quad +\left(\sum_{q=-\infty}^{\infty} |J_q(z)| x^{-q} \sum_{|n|>M}  |J_n(z)| x^n \right)^2 \nn
& \quad \left. + \sum_{q=-\infty}^{\infty} |J_q(z)| x^{-q} \sum_{|n|>M}  |J_n(z)| x^n \sum_{|m|>M}  |J_m(z)| x^{-m} \sum_{p=-M}^{M} |J_p(z)| x^{p} \right] \nn
&\le \sum_{q=-\infty}^{\infty} |J_q(z)| x^{q} \sum_{|n|>M} |J_n(z)| x^n +\left(\sum_{q=-\infty}^{\infty} |J_q(z)| x^{q} \sum_{|n|>M}  |J_n(z)| x^n \right)^2 .
\label{eq:shreal}
\end{align}
We will take $x=M/(z \zeta)$, where $\zeta\approx 1.8$ is the solution of $e^{1+1/2\zeta}=2\zeta$.
This number can be obtained as $1/(2\; {\rm ProductLog}[1/e])$ in Mathematica.
It is then easy to check numerically that the right-hand side (RHS) of \eqref{eq:shreal} is no greater than $1/2$.

To address this bound analytically,
we can use the upper bound
\begin{equation}
|J_m(z)|\le \frac 1{|m|!}\left|\frac z2\right|^{|m|} .
\end{equation}
Using this upper bound gives
\begin{equation}
\label{sumco}
\sum_{m=-\infty}^\infty |J_m(z)| x^m \le e^{xz/2} + (e^{z/2x}-1) \, ,
\end{equation}
and
\begin{equation}
\label{difco}
\sum_{|m|>M} |J_m(z)| x^m \le 2\frac{(xz/2)^{M+1}}{(M+1)!}+2\frac{(z/2x)^{M+1}}{(M+1)!}\, ,
\end{equation}
provided that $xz\le M+2$ and $z/x\le M+2$.
We take $x=M/(z \zeta)$, in which case the inequalities are satisfied.
For $M>2$ the second terms in Eqs.~\eqref{sumco} and \eqref{difco} are much smaller and may be omitted.
Keeping the lowest-order term in the expansion of $\ser_+(x)$, we obtain
\begin{equation}
\ser_+(x) \lsim 2\frac{(xz/2)^{M+1}}{(M+1)!}e^{xz/2}\, .
\end{equation}
Using $x=M/(z \zeta)$ and Stirling's approximation, it can be shown that the RHS decreases with $M$, and is $< 1/2$.
This part of the derivation is for $M>2$, and for the case $M=2$ it is easily checked numerically that the RHS is no greater than $1/2$.
The choice of $\zeta$ is so that the asymptotic expression decreases with $M$.
Hence $1-\ser_+(x) \gsim 1/2$, so Eq.~\eqref{crpbnd} gives $\corrp(x)\lsim 2^{r}$.
Using this bound in Eq.~\eqref{trick} gives Eq.~\eqref{errbnd}, as required.}
\end{proof}

If we want the correction to give no more than a factor of $2$ to the complexity, we can choose $N=3\segs M$, so
\begin{equation}
\sum_{|m|>N} |a_m| \lsim {2^{\segs+1}}{\left(\frac{z\zeta}{M}\right)^{3\segs M+1}}.
\end{equation}
Since this sum of terms gives the order of the error when the correction is used, we have an error from this correction that is exponentially small, as expected.
If we require higher accuracy, we could increase $N$ so there are more terms in the correction.
However, it is also possible to obtain higher accuracy by increasing $M$.
Since that does not increase the complexity (except for a factor of $2$) until $N=3\segs M$, it is advantageous to simply take $N=3\segs M$.

Next we can prove the overall result for the simulation with the correction.
\begin{proof}[Proof of \thm{corthm}.]
To prove this Theorem we use Algorithm \ref{alg1} and select $N=3\segs M$.
Algorithm \ref{alg1} proceeds by performing $\segs$ segments as in step 2, where each segment uses a superposition of powers of the step of the quantum walk with a cutoff of $M$ together with \obv.
Then the correction in step 3 uses a sum over powers of $U$ truncated at $N=3\segs M$, with weightings $a_m$.
Then {\obv} on the entire operation is used in step 4 to obtain success probability near $1$.

The complexity in this Theorem is obtained from two requirements.
First, we require $M$ to be sufficiently large that $s(\appcor)$ is not larger than $2$, so we can perform the final {\obv} with a single step.
Second, we require that the error at the end is no greater than $\epsilon$.
After the correction, the state is of the form \eqref{aftercor}, where $s=s(\appcor)$, and
\begin{equation}
s(\appcor) = \sum_{m=-N}^N |a_m|\, .
\end{equation}
That is, there is a $1/s(\appcor)$ factor on the amplitude for success, which is flagged by zeros in the ancilla qubits.

Using the result of \lem{sumbnd}, there is the upper bound
\begin{align}
s(\appcor) &\lsim \left( 1-2\sum_{|m|>M}|J_m(z)| \right)^{-\segs} \nn
&\approx \exp\left( 2\segs \sum_{|m|>M}|J_m(z)| \right).
\end{align}
Therefore we may obtain $s(\appcor)\lsim 2$ for
\begin{equation}\label{ine1}
\segs \lsim \frac {\log 2}{2\sum_{|m|>M}|J_m(z)|}\, .
\end{equation}
Using the inequality
\begin{equation}\label{ine2}
\frac {\log 2}{2\sum_{|m|>M}|J_m(z)|} \gsim \frac {(M+1)!\log 2}{8|z/2|^{M+1}}\, ,
\end{equation}
we find that provided
\begin{equation}\label{ine3}
\segs \lsim \frac {(M+1)!\log 2}{8|z/2|^{M+1}}
\end{equation}
is satisfied, we still obtain $s(\appcor)\lsim 2$.
This is because Eq.~\eqref{ine3} together with Eq.~\eqref{ine2} implies Eq.~\eqref{ine1}.
In order to obtain a simulation for overall time $t$, we require $\segs = \Theta(tXd)$, where $X$ can be taken equal to $\|H\|_{\max}$, and therefore $\segs=\Theta(\tau)$.
According to Eq.~\eqref{ine3}, we can obtain $s(\appcor)\lsim 2$ with $M=\Theta(\log \tau/\log\log \tau)$.

Next we consider the criterion that the error be no larger than $\epsilon$, which may possibly require a larger value of $M$.
Using \lem{corerr}, we have
\begin{equation}
\sum_{|m|>N} |a_m| \lsim {2^{\segs+1}}{\left(\frac{z\zeta}{M}\right)^{3\segs M+1}}.
\end{equation}
This expression gives the scaling in the error after the correction.
Using Eq.~\eqref{ampamp}, the error after {\obv} is of the same order.

Now, if the RHS is less than the allowable error $\epsilon$ with $M=\Theta(\log \tau/\log\log \tau)$, then we immediately have the error appropriately bounded with complexity $O(\tau \log \tau/\log\log \tau)$.
Otherwise we can increase $M$, and still obtain $s(\appcor)\lsim 2$.
Taking $3\segs M = O(\log(1/\epsilon))$ is sufficient to obtain error no larger than $\epsilon$.
Since $3\segs M$ is the same as the complexity up to a multiplicative factor, the overall result for the complexity is
\begin{equation}
O\left( \tau \frac{\log\tau}{\log\log\tau} + \log(1/\epsilon) \right).
\end{equation}
\end{proof}
\section{Two Rounds of Error Correction}
\label{sec:double}
Next we consider the result with a second round of correction, as in Algorithm \ref{alg2}.
That is, we combine the multiple segments followed by a correction from the previous section into a compound segment; we repeat this compound segment multiple times, then perform a correction.
This technique enables the scaling in the following Theorem.
\begin{theorem}
\label{thm:corthm2}
A $d$-sparse Hamiltonian $H$ acting on $n$ qubits can be simulated for time $t$ within error $\epsilon>0$ with
\begin{equation}\label{recursecomplexity}
O\left(\tau \frac{\log\log\tau}{\log\log\log\tau}+\log(1/\epsilon) \right)
\end{equation}
queries, where $\tau := t\|H\|_{\max}d$.
\end{theorem}
In order to show this result, we again need to show a number of intermediate results.
The first is a bound on the magnitudes of the individual coefficients $a_m$, rather than the sum.
\begin{lemma}\label{lem:corerr2}
The coefficients $a_m$ for the correction operator $\corr(U)$ satisfy
\begin{equation}
\label{abnd}
|a_m| \lsim {2^\segs}{\left(\frac{z\zeta}{M}\right)^{|m|}},
\end{equation}
where $\zeta\approx 1.8$ is the solution of $e^{1+1/2\zeta}=2\zeta$.
\end{lemma}
\begin{proof}{
Similar to Eq.~\eqref{trick}, for $x\ge 1$,
\begin{align}\label{tricky}
|a_m| & \le \sum_{n=|m|}^{\infty}|a_n| \nn
&\le \frac 1{x^{|m|}}\sum_{n=|m|}^{\infty} |a_n| x^n \nn
&\le \frac 1{x^{|m|}}\corrp(x) \,.
\end{align}
In exactly the same way as in the proof of \lem{corerr}, we can take $x=M/(z \zeta)$, to obtain $\ser_+(x) \lsim 1/2$ and $\corrp(x)\le 2^\segs$.
We then obtain Eq.~\eqref{abnd} as required.}
\end{proof}

This result shows that the $|a_m|$ exponentially decrease, so we can use this to show that the sum over $|a_m|$ is properly bounded, and we can even bound a sum over $|a_m|x^m$.
A key step in the proof for a single round of error correction is Eq.~\eqref{difco}, which bounds the difference of the coefficients of the ideal sequence and the sequence we have just performed.
Therefore, if we can bound that for the corrected series, then we can show that the recursion works.

The particular result we find is as in the following Lemma.
\begin{lemma}\label{lem:sumbnd2}
Consider Algorithm \ref{alg2}, with $N=3M\segs$, and $M$ chosen such that 
\begin{equation}\label{rest}
\segs \le \frac{\log 2}{2 \sum_{|m|>M} |J_m(z)|} \, .
\end{equation}
The coefficients $a'_m$ for the correction $\corr'(U)$ satisfy
\begin{equation}
\label{abnd2}
\sum_{m=-\infty}^{\infty}|a'_m| \lsim \left\{ 1-2\sum_{|m|>N}|a_m| \right\}^{-\segs'}.
\end{equation}
\end{lemma}
\begin{proof}{
First we define operators for what is achieved with $\segs$ steps and correction for the compound segment, in a similar way as we did for the first round of correction.
We use primed variables to indicate the new variable names corresponding to those for the first round of correction.
The exact operation, if there were perfect correction on the compound segment, is denoted $V'$, and the actual operation performed is denoted $\appseg'$.
The difference is denoted $\Delta'$, and $\Delta'=\appseg'-V'$.
As before these are all functions of the step operator $U$.
These functions may be written as
\begin{align}
V' &= \corr \oaa^\segs \nn
\appseg' &= \appcor \oaa^\segs \nn
\Delta' &= \left(\corr-\appcor \right) \oaa^\segs \, .
\end{align}
Following exactly the same derivation for these primed quantities as in the proof of \lem{sumbnd}, we have
\begin{equation}\label{cor2exp}
\corr' = \sum_{k=0}^{\infty}\binom{\segs+k-1}{\segs-1} \ser'^k,
\end{equation}
where
\begin{equation}
\ser' =\frac 12 \left(V'^\dagger\Delta'-\Delta'^\dagger \appseg'+ V'^\dagger\Delta' V'^\dagger\Delta'+V'^\dagger\Delta'\Delta'^\dagger \appseg'\right).
\end{equation}

We can rewrite $W'$ as
\begin{equation}\label{serp}
\ser' =\frac 12 (\oaa^\dagger \oaa)^\segs\left(\corr^\dagger\Delta_C-\Delta_C^\dagger \appcor\right)
+\frac 12 (\oaa^\dagger \oaa)^{2\segs} \left(\corr^\dagger\Delta_C \corr^\dagger\Delta_C+\corr^\dagger\Delta_C\Delta_C^\dagger \appcor \right),
\end{equation}
where $\Delta_C:=\corr-\appcor$.
Next, $\oaa^\dagger \oaa$ can be expanded as
\begin{align}\label{unicor}
\oaa^\dagger \oaa &= \left[\frac 32 \appseg^\dagger - \frac 12 \appseg^\dagger \appseg \appseg^\dagger \right]\left[ \frac 32\appseg - \frac 12\appseg\appseg^\dagger \appseg \right] \nn
&= \frac 94 \appseg^\dagger\appseg - \frac 32 (\appseg^\dagger\appseg)^2 + \frac 14 (\appseg^\dagger\appseg)^3 \nn
&= \openone - \frac 32 \Delta^\dagger \Delta + \frac 34 (\Delta^\dagger \Delta)^2 - \frac 14 (\Delta^\dagger \Delta)^3 - \frac 34 \left( (\appseg^\dagger \Delta)^2 +  (\Delta^\dagger\appseg)^2 \right) \nn
&\quad  -\frac 14 \left( (\appseg^\dagger \Delta)^3 + (\Delta^\dagger \appseg)^3\right) - \frac 34 \Delta^\dagger\Delta\left( \Delta^\dagger \appseg+  \appseg^\dagger \Delta\right).
\end{align}
Therefore, using the properties of the functional $s$ and $s(\appseg)\le 2$,
\begin{align}
s(\oaa^\dagger \oaa) &\le 1 + \frac 32 [s(\Delta)]^2 + \frac 34 [s(\Delta)]^4 + \frac 14 [s(\Delta)]^6 + 6[s(\Delta)]^2 + 4[s(\Delta)]^3 + 3[s(\Delta)]^3 \nn
&= 1+\frac {15}2  [s(\Delta)]^2 + 7 [s(\Delta)]^3 + \frac {3}4  [s(\Delta)]^4 + \frac 14  [s(\Delta)]^6.
\end{align}
Given that we have chosen $M$ so that Eq.~\eqref{Mval} holds, we have
\begin{align}
s\left((\oaa^\dagger \oaa)^\segs\right) &\le \left(1+\frac {15}2  [s(\Delta)]^2+O\left( [s(\Delta)]^3 \right)\right)^{\frac{\log 2}{2s(\Delta)}} \nn
&= 1+\frac{15}4 (\log 2) s(\Delta) + O\left( [s(\Delta)]^2 \right).
\end{align}
In addition, the restriction \eqref{rest} ensures that $s(\corr)\le 2$.
Hence we have
\begin{equation}
s(\ser') \le 2s(\Delta_C) + 4[s(\Delta_C)]^2 + O\left( [s(\Delta)]^2s(\Delta_C) \right).
\end{equation}
Using Eq.~\eqref{cor2exp} we then have
\begin{equation}
s(\corr') \lsim \left[ 1- 2s(\Delta_C)\right]^{-\segs'}.
\end{equation}
Recognising that $s(\Delta_C)\le \sum_{|m|>N} |a_m|$, this gives the result required.}
\end{proof}

Using Lemma \ref{lem:corerr}, and taking $N=3\segs M$, we then find that we can take $\segs'$ satisfying
\begin{equation}\label{muval}
\segs' \lsim \frac {\log 2}{2^{\segs+2}}{\left(\frac{M}{z\zeta}\right)^{3\segs M+1}},
\end{equation}
and obtain $s(\corr')\le 2$.
This expression is the equivalent of Eq.~\eqref{ine3}.
Next we prove a Lemma bounding the size of the error.
\begin{lemma}\label{lem:corerr3}
Consider Algorithm \ref{alg2}, with $N=3M\segs$, and $M$ chosen such that 
\begin{equation}
\segs \le \frac{\log 2}{2 \sum_{|m|>M} |J_m(z)|}\, .
\end{equation}
The coefficients $a'_m$ for the correction $\corr'(U)$ satisfy
\begin{equation}
\label{errbnd3}
\sum_{|m|>N'} |a'_m| \lsim {2^{\segs'}}{\left(\frac{z \zeta\zeta' 2^{1/M}}{M}\right)^{N'+1}},
\end{equation}
where $\zeta\approx 1.8$ is the solution of $e^{1+1/2\zeta}=2\zeta$ and $\zeta'\approx 1.5$ is a solution of $\zeta'^5(\sqrt 2-2\zeta')^2=16\sqrt 2$.
\end{lemma}
\begin{proof}{
In the same way as for \lem{corerr}, we have
\begin{equation}\label{trick2}
\sum_{|m|>N'} |a'_m| \le \frac 2{x^{N'+1}}{\corrp}'(x)\, ,
\end{equation}
where ${\corrp}'(x):=\sum_{m=-\infty}^{\infty} |a'_m| x^m$.
This function satisfies
\begin{equation}\label{crpbnd2}
{\corrp}'(x) \le \left[ 1-\ser'_+(x)\right]^{-\segs'} ,
\end{equation}
where $\ser'_+$ is the function $\ser'$ modified to take the absolute values of all coefficients.

To bound $\ser'_+(x)$, we can use Eq.~\eqref{serp} and Eq.~\eqref{unicor}.
For $\oaa^\dagger \oaa$ with the absolute values of all coefficients taken, it will be upper bounded by
\begin{equation}
\label{boundbb}
1+ \frac 32 \delta^2 + \frac 34 \delta^4 + \frac 14 \delta^6+\frac 32 \nu^2\delta^2+\frac 12 \nu^3\delta^3+\frac 32 \nu \delta^3,
\end{equation}
where
\begin{align}
\nu &:= \sum_{q=-M}^{M} |J_q(z)| x^{q} , \\
\delta & := \sum_{|n|>M} |J_n(z)| x^n.
\end{align}
These expressions are upper bounded in Eqs \eqref{sumco} and \eqref{difco}.
It was found that for $x\le M/(z\zeta)$, $\nu\delta<1/2$.
In addition, for this choice of $x$, $\nu\gg 1$ and $\delta\ll 1$, so the overall value will be no greater than $2$.
Then, using \eqref{serp}, we obtain
\begin{equation}
\ser'_+(x) \le 2^\segs \sum_{q=-\infty}^{\infty} |a_q| x^q \sum_{|n|>N} |a_n| x^n +\left( 2^\segs \sum_{q=-\infty}^{\infty} |a_q| x^q \sum_{|n|>N} |a_n| x^n \right)^2 .
\end{equation}
Note that we can use $|a_m|=|a_{-m}|$ due to symmetry.
Then, using the bound \eqref{abnd}, we have
\begin{equation}
\sum_{|n|>N} |a_n| x^n \le \frac{2^{\segs+1}}{1-z\zeta x/M}\left(\frac{z\zeta x}M\right)^{N+1}.
\end{equation}
Similarly the sum over all powers can be bounded as
\begin{equation}
\sum_{q=-\infty}^{\infty} |a_q| x^q \le \frac{2^{\segs+1}}{1-z\zeta x/M}\, .
\end{equation}
Therefore we have
\begin{equation}
\sum_{|m|>N'} |a'_m| \lsim \frac 1{x^{N'+1}} \left\{ 1- 4 \frac{2^{3\segs}}{(1-z\zeta x/M)^2}\left(\frac{z\zeta x}M\right)^{N+1} \right\}^{-\segs'}.
\end{equation}
We now wish to take $x$ to be slightly less than $M/(z \zeta 2^{1/M})$, so that the expression in braces is $\ge 1/2$.
In particular we take $x=M/(z \zeta \zeta' 2^{1/M})$, where $\zeta'\approx 1.52937$ is a solution of $\zeta'^5(\sqrt 2-2\zeta')^2=16\sqrt 2$.
Then we obtain
\begin{equation}
\sum_{|m|>N'} |a'_m| \lsim \left(\frac{z \zeta\zeta' 2^{1/M}}M\right)^{N'+1} 2^{\segs'}.
\end{equation}}
\end{proof}

We are now in a position to prove the Theorem for the complexity.
\begin{proof}[Proof of \thm{corthm2}.]
For this proof we use Algorithm \ref{alg2}.
The simulation proceeds by using compound segments, where each segment uses $\segs$ segments and a correction.
We perform $\segs'$ of these compound segments, followed by an overall correction. 
Now the overall length of the simulation is $\segs\segs'$, so we require $\segs\segs'=\tau$.
We have three requirements:
\begin{enumerate}
\item The corrections for the compound segments satisfy $s(\corr)\le 2$, so {\obv} can be performed in one step.
\item The final correction satisfies $s(\corr')\le 2$, so the final {\obv} can be performed in one step.
\item The error as obtained in \lem{corerr3} is upper bounded by $\epsilon$.
\end{enumerate}

Considering the first requirement, let us take
\begin{equation}
M=\Theta\left( \frac{\log\log\tau}{\log\log\log\tau} \right).
\end{equation}
Then we have Eq.~\eqref{Mval} satisfied for $\segs=\Theta(\log\tau)$, which implies that $s(\corr)\le 2$.
Second, we find that we have Eq.~\eqref{muval} satisfied with $\segs'=\Theta(\tau/\segs)$.
Then we obtain $s(\corr')\le 2$ for the final {\obv}, and $\segs\segs'=\Theta(\tau)$ as required.
Because the complexity of step 2 of Algorithm \ref{alg2} is $\segs\segs' M$ up to multiplying factors, it gives a contribution to the complexity of
\begin{equation}
\Theta\left(\tau \frac{\log\log\tau}{\log\log\log\tau} \right) .
\end{equation}

Finally we consider the third requirement.
We can choose $N'=9\segs\segs'M$ without changing the complexity.
Now if the expression on the RHS of Eq.~\eqref{errbnd3} in \lem{corerr3} is less than $\epsilon$, then we have satisfied this requirement.
Otherwise we can further increase $M$ in order to obtain error no greater than $\epsilon$.
The overall complexity is equal to $N'$ up to multiplying factors.
We can obtain the RHS of Eq.~\eqref{errbnd3} less than $\epsilon$ by taking $N'=O(\log(1/\epsilon))$.
Hence the overall complexity sufficient to satisfy both requirements is as given in Eq.~\eqref{recursecomplexity}.
\end{proof}

\section{Conclusions}
\label{conc}
We have shown how to perform corrections on the superposition of quantum walk steps approach to Hamiltonian simulation from Ref.~\cite{BCK15}.
This approach gives a result much closer to the lower bound for complexity than Ref.~\cite{BCK15}, because it has a sum rather than a product in the scaling of the complexity.
Our result is very close to the lower bound for the complexity in Ref.~\cite{BCK15}, except our result differs by double-logarithmic factors for the scaling in $\tau$ and $\epsilon$.

Our approach to correcting the quantum walk is sufficiently flexible that it can be used to perform an arbitrary number of rounds of correction.
That should provide scaling of the complexity with further iterated logarithms of $\tau$, though not strictly linear scaling in $\tau$.
Proving the complexity scaling is quite complicated, so we have limited to analyzing two rounds of correction here.

After completion of this work, we were made aware of another approach that yields complexity closer to the lower bound \cite{low}.
On the other hand, our result is more flexible than that in Ref.~\cite{low}, because it can be applied to simulations with sums of many different unitary operators.
In contrast, the method of Ref.~\cite{low} only applies to using a single unitary operator to simulate Hamiltonian evolution.
This means that our approach can be applied to both simulations based on quantum walks, and simulations based on a Taylor series \cite{BCK15}, whereas the method of Ref.~\cite{low} does not apply to Taylor series.
That is another important case, because it seems more useful for quantum chemistry, and also has better scaling in the number of additional gates.
In addition, it is needed for Hamiltonians that are a sum of operators sparse in different bases, which the quantum walk approach cannot be applied to.
For example, particles in a potential have a Hamiltonian of this type \cite{rolando}.

\section*{Acknowledgements}
We thank Aaron Ostrander, Aram Harrow and Andrew Childs for valuable discussions.
We acknowledge support from IARPA contract number D15PC00242.
DWB is funded by an Australian Research Council Future Fellowship (FT100100761) and a Discovery Project (DP160102426).
LN thanks the support from Funda\c{c}\~{a}o para a Ci\^{e}ncia e a Tecnologia (Portugal), namely through programmes PTDC/POPH/POCH and projects UID/EEA/50008/2013, IT/QuSim, ProQuNet, partially funded by EU FEDER, and from the EU FP7 project PAPETS (GA 323901). Furthermore, LN acknowledges the support from the DP-PMI and FCT (Portugal) through scholarship SFRH/BD/52241/2013.

\end{document}